\documentclass[11pt]{article}

\pdfoutput=1

\usepackage{a4}
\usepackage[top=30truemm,bottom=30truemm]{geometry}

\usepackage[running]{lineno}

\usepackage{authblk}
\usepackage{amsmath}
\usepackage{amssymb}
\usepackage{amsfonts}
\usepackage{amsthm}
\usepackage{cases}
\usepackage{url}
\usepackage{graphicx}

\usepackage{url}
\usepackage{multirow}
\usepackage{multicol}
\usepackage{cases}

\newcommand{\prev}{\mathsf{prev}}

\newcommand{\occ}{\mathit{occ}}

\newcommand{\spe}{\mathsf{spe}}

\newcommand{\PSTray}{\mathsf{PSTray}}
\newcommand{\PSTree}{\mathsf{PSTree}}
\newcommand{\PSA}{\mathsf{PSA}}
\newcommand{\PLCP}{\mathsf{PLCP}}

\newcommand{\fpos}{\mathsf{fpos}}
\newcommand{\f}{\mathsf{f}}

\newcommand{\rank}{\mathsf{rank}}

\newtheorem{theorem}{Theorem}
\newtheorem{lemma}{Lemma}

\theoremstyle{definition}
\newtheorem{definition}{Definition}

\title{The Parameterized Suffix Tray}

\author{Noriki~Fujisato$^{1}$}
\author{Yuto~Nakashima$^{1}$}
\author{Shunsuke~Inenaga$^{1,2}$}
\author{Hideo~Bannai$^3$}
\author{Masayuki~Takeda$^{1}$}

\affil{
  \normalsize{
  \textit{$^1$Department of Informatics, Kyushu University, Fukuoka, Japan}\\
  \texttt{\{noriki.fujisato, yuto.nakashima, inenaga, takeda\}@inf.kyushu-u.ac.jp}\\
  \textit{$^2$PRESTO, Japan Science and Technology Agency, Kawaguchi, Japan}\\
  \textit{$^3$M\&D Data Science Center, Tokyo Medical and Dental University, Tokyo, Japan}\\
  \texttt{hdbn.dsc@tmd.ac.jp}
  }
}

\date{}

\begin{document}
\maketitle

\begin{abstract}
  Let $\Sigma$ and $\Pi$ be disjoint alphabets,
  respectively called the static alphabet and the parameterized alphabet.
  Two strings $x$ and $y$ over $\Sigma \cup \Pi$
  of equal length are said to \emph{parameterized match} (\emph{p-match})
  if there exists a renaming bijection $f$ on $\Sigma$ and $\Pi$
  which is identity on $\Sigma$ and maps the characters of $x$ to those
  of $y$ so that the two strings become identical.
  The indexing version of the problem of finding p-matching occurrences
  of a given pattern in the text is a well-studied topic in string matching.
  In this paper, we present a state-of-the-art indexing structure for p-matching
  called the \emph{parameterized suffix tray} of an input text $T$,
  denoted by $\PSTray(T)$.
  We show that $\PSTray(T)$
  occupies $O(n)$ space and supports pattern matching queries in $O(m + \log (\sigma+\pi) + \occ)$ time,
  where $n$ is the length of $t$, $m$ is the length of a query pattern $P$, $\pi$ is the number of distinct symbols of $|\Pi|$ in $T$,
  $\sigma$ is the number of distinct symbols of $|\Sigma|$ in $T$
  and $\occ$ is the number of p-matching occurrences of $P$ in $T$.
  We also present how to build $\PSTray(T)$ in $O(n)$ time
  from the parameterized suffix tree of $T$.
\end{abstract}

\section{Introduction}

\emph{Parameterized Pattern Matching} (PPM),
first introduced by Baker~\cite{Baker96} in 1990's,
is a well-studied class of pattern matching
motivated by plagiarism detection, software maintenance,
and RNA structural matching~\cite{Baker96,Shibuya04,MendivelsoTP20}.

PPM is defined as follows:
Let $\Sigma$ and $\Pi$ be disjoint alphabets.
Two equal-length strings $x$ and $y$ from $\Sigma \cup \Pi$ are said to parameterized match (p-match) if $x$ can be transformed to $y$ by applying a bijection which renames the elements of $\Pi$ in $x$ (the elements of $\Sigma$ in $x$ must remain unchanged).
PPM is to report every substring in a text $T$ that p-matches a pattern $P$.

In particular, the indexing version of PPM,
where the task is to preprocess an input text string $T$
so that parameterized occurrences of $P$ in $T$ can be reported quickly,
has attracted much attention for more than two decades
since the seminal paper by Baker~\cite{Baker96}.

Basically, the existing indexing structures for p-matching are
designed upon indexing structure for exact pattern matching.
Namely,
\emph{parameterized suffix trees}~\cite{Baker96},
\emph{parameterized suffix arrays}~\cite{DeguchiHBIT08},
\emph{parameterized DAWGs}~\cite{NakashimaFHNYIB20},
\emph{parameterized CDAWGs}~\cite{NakashimaFHNYIB20},
\emph{parameterized position heaps}~\cite{Kucherov13,FujisatoNIBT18},
and \emph{parameterized BWTs}~\cite{DBLP:conf/soda/0002ST17}
are based on
their exact matching counterparts:
suffix trees~\cite{Weiner},
suffix arrays~\cite{manber93:_suffix},
DAWGs~\cite{blumer85:_small_autom_recog_subwor_text},
CDAWGs~\cite{Blumer87},
position heaps~\cite{Kucherov13,ehrenfeucht_position_heaps_2011},
and BWTs~\cite{BWT}, respectively.
It should be emphasized that extending exact-matching indexing structures
to parameterized matching is not straightforward and poses algorithmic challenges.
Let $n$, $m$, $\pi$ and $\sigma$ be the lengths of a text $T$, a pattern $P$, the number of distinct symbols of $|\Pi|$ that appear in $T$ and the number of distinct symbols of $|\Sigma|$ that appear in $T$, respectively.
While there exist a number of algorithms which construct
the suffix array for $T$ in $O(n)$ time in the case of integer alphabets
of polynomial size in $n$~\cite{Farach-ColtonFM00,KarkkainenSB06,KoA05,KimSPP05,NongZC11,Baier16},
the best known algorithms build
the parameterized suffix array for $T$ (denoted $\PSA(T)$)
in $O(n \log(\sigma+\pi))$ time via the suffix tree~\cite{Baker96,Shibuya04},
or directly in $O(n\pi)$ time~\cite{FujisatoNIBT19}.
The existence of a pure linear-time algorithm for building
$\PSA(T)$ and the parameterized suffix tree (denoted $\PSTree(T)$)
in the case of integer alphabets remains open.

PPM queries can be supported
in $O(m+\log n + \occ)$ time by
$\PSA(T)$ coupled with the parameterized LCP array (denoted $\PLCP(T)$)~\cite{DeguchiHBIT08},
or in $O(m\log (\sigma+\pi)+\occ)$ time by $\PSTree(T)$~\cite{Baker96},
where $\occ$ is the number of occurrences to report.

In this paper, we propose a new indexing structure for p-matching,
the \emph{parameterized suffix tray} for $T$ (denoted $\PSTray(T)$).
$\PSTray(T)$ is a combination of $\PSTree(T)$ and $\PSA(T)$
and is an analogue to the \emph{suffix tray} indexing structure
for exact matching~\cite{ColeTL2015}.
We show that our $\PSTray(T)$
\begin{itemize}
  \item[(1)] occupies $O(n)$ space,
  \item[(2)] supports PPM queries in $O(m+\log(\sigma+\pi)+\occ)$, and
  \item[(3)] can be constructed in $O(n)$ time from $\PSTree(T)$ and $\PSA(T)$.
\end{itemize}
Result (3) implies that $\PSTray(T)$ can be constructed in
$O(n \min\{\log(\sigma + \pi), \pi\})$ time using $O(n)$ working space~\cite{Baker96,Shibuya04,FujisatoNIBT19}.
Results (1) and (2) together with this
imply that our $\PSTray(T)$ is the fastest linear-space
indexing structure for PPM which can be built in time linear in $n$.

We emphasize that extending suffix trays
for exact matching~\cite{ColeTL2015} to parameterized matching is also not
straightforward.
The suffix tray of a string $T \in \Sigma^*$ is a hybrid data structure of the
suffix tree and suffix array of $T$, designed as follows:
Each of the $O(\frac{n}{\sigma})$ carefully-selected nodes of the suffix tree
stores an array of fixed size $\sigma$,
so that pattern traversals within these selected nodes take $O(m)$ time
(this also ensures a total space to be $O(\frac{n}{\sigma} \times \sigma) = O(n)$).
Once the pattern traversal reaches an unselected node,
then the search switches to the sub-array of the suffix array of size $O(\sigma)$.
This ensures a worst-case $O(m + \log \sigma + \occ)$-time pattern matching
with the suffix tray.

Now, recall that the previous-encoded suffixes of $T$ are
sequences over an alphabet $\Sigma \cup \{0, \ldots, n-1\}$
of size $\Theta(\sigma + n) \subseteq O(n)$,
while the alphabet size of $T$ is $\sigma + \pi$.
This means that na\"ive extensions of suffix trays to PPM
would only result in either super-linear $O(\frac{n^2}{\sigma + \pi})$ space,
or $O(m + \log n + \occ)$ query time which can be achieved
already with the parameterized suffix array.
We overcome this difficulty by using
the \emph{smallest parameterized encoding} (\emph{spe}) of strings
which was previously proposed by the authors in the context of
PPM on labeled trees~\cite{FujisatoNIBT19b},
and this leads to our $O(n)$-space parameterized suffix trays
with desired $O(m + \log(\sigma + \pi) + \occ)$ query time.

\section{Preliminaries}

Let $\Sigma$ and $\Pi$ be disjoint ordered sets of characters,
respectively called the static alphabet and the parameterized alphabet.
We assume that any character in $\Pi$ is lexicographically smaller
than any character in $\Sigma$.
An element of $(\Sigma \cup \Pi)^*$ is called a \emph{p-string}.
For a (p-)string $w = xyz$, $x$, $y$ and $z$ are called a \emph{prefix}, \emph{substring}, and \emph{suffix} of $w$.
The $i$-th character of a (p-)string $w$ is denoted by $w[i]$ for $1 \leq i \leq |w|$,
and the substring of a (p-)string $w$ that begins at position $i$ and ends at position $j$ is denoted by $w[i:j]$ for $1 \leq i \leq j \leq |w|$. For convenience, let $w[i:j] = \varepsilon$ if $j < i$.
Also, let $w[i:] = w[i:|w|]$ for any $1 \leq i \leq |w|$,
and $w[:j] = w[1:j]$ for any $1 \leq j \leq |w|$.
For any (p-)string $w$, let $w^R$ denote the reversed string of $w$.
If a p-string $x$ is lexicographically smaller than a p-string $y$,
then we write $x<y$.

\begin{definition}[Parameterized match~\cite{Baker93}]
  Two p-strings $x$ and $y$ of the same length are said to \emph{parameterized match} (\emph{p-match}) iff there is a bijection $f$ on $\Sigma \cup \Pi$ such that
  $f(c) = c$ for any $c \in \Sigma$ and $x[i]=f(y[i])$ for any $1\leq i\leq |x|$.
\end{definition}

We write $x \approx y$ iff two p-strings $x,y$ p-match. For instance, if $\Sigma=\{\mathtt{A},\mathtt{B}\}$, $\Pi=\{\mathtt{x},\mathtt{y},\mathtt{z}\}$, then $X=\mathtt{xyzAxxxByzz}$ and $Y=\mathtt{zxyAzzzBxyy}$ p-match since there is a bijection $f$ such that $f(\mathtt{A})=\mathtt{A}$, $f(\mathtt{B})=\mathtt{B}$, $f(\mathtt{x})=\mathtt{z}$, $f(\mathtt{y})=\mathtt{x}$, $f(\mathtt{z})=\mathtt{y}$, and $$f(\mathtt{x})f(\mathtt{y})f(\mathtt{z})f(\mathtt{A})f(\mathtt{x})f(\mathtt{x})f(\mathtt{x})f(\mathtt{B})f(\mathtt{y})f(\mathtt{z})f(\mathtt{z})=\mathtt{zxyAzzzBxyy}=Y.$$

\begin{definition}[Parameterized Pattern Matching problem(PPM)~\cite{Baker93}]
  Given a text p-string $T$ and a pattern p-string $P$, find all positions $i$ in $T$ such that $T[i:i+|P|-1]\approx p$.
\end{definition}

For instance, if $\Sigma=\{\mathtt{A}\}$, $\Pi=\{\mathtt{x},\mathtt{y}$, $\mathtt{z}\}$, $T=\mathtt{xyzAxxxAyyzAzx}$, and $P=\mathtt{yAzz}$, then the out put for PPM is $\{3,7\}$.
We call the positions in the output of PPM the \emph{p-beginning positions}
for given text $T$ and pattern $P$. We say that the pattern \emph{p-appears} in the text $T$ iff the pattern and a substring of the text p-match. In this paper, we suppose that a given text $T$ terminates with a special end-marker $\$$ which occurs nowhere else in $T$. We assume that $\$$ is an element of $\Sigma$ and $\$$ is lexicographically larger than any elements from $\Sigma$ and $\Pi$.

\begin{definition}[Previous encoding~\cite{Baker93}]
  \label{prev}
  For a p-string $w$, the previous encoding $\prev(w)$ is a string of length $|w|$ such that for each $1 \leq i \leq |w|$,
  \[
    \prev(w)[i] =
    \begin{cases}
      w[i] & \mbox{if $w[i]\in\Sigma$},                                        \\
      0    & \mbox{if $w[i]\in\Pi$ and $w[j]\neq w[i]$ for any $1\leq j < i$}, \\
      i-j  & \mbox{otherwise, $w[i]=w[j]$ and $w[i]\neq w[k]$ for any $j<k<i$}.
    \end{cases}
  \]
\end{definition}
Intuitively, when we transform $w$ to $\prev(w)$,
the first occurrence of each element of $\Pi$ is replaced with 0 and any other occurrence of the element of $\Pi$ is replaced by the distance to the previous occurrence of the same character, and each element of $\Sigma$ remains the same.

\begin{definition}[Smallest parameterized encoding (spe)~\cite{FujisatoNIBT19b}]
  \label{spe}
  For a p-string $w$, the smallest parameterized encoding $\spe(w)$ is
  the lexicographically smallest p-string such that $w \approx \spe(w)$.
\end{definition}
Namely, $\spe(w)$ maps a given string $w$
to the representative of the equivalence class of p-strings
under p-matching $\approx$.

For any two p-strings $w_1,w_2$, $\prev(w_1)= \prev(w_2)\Leftrightarrow \spe(w_1)= \spe(w_2)\Leftrightarrow w_1\approx w_2$.
For instance, let $\Sigma=\{\mathtt{A,B}\}$, $\Pi=\{\mathtt{x},\mathtt{y},\mathtt{z}\}$, $X=\mathtt{yxzAyyyBxzz}$, and $Y=\mathtt{zxyAzzzBxyy}$. Then $\prev(X)=\mathtt{000A411B771}=\prev(Y)$ and $\spe(X)=\mathtt{xyzAxxxByzz}=\spe(Y)$.

\section{Parameterized suffix trays}

In this section, we propose a new indexing structure called the \emph{parameterized suffix tray} for PPM, and we discuss its space requirements.

Our parameterized suffix trays are a ``hybrid'' data structure
of parameterized suffix trees and parameterized suffix arrays,
which are defined as follows:

\begin{definition}[Parameterized suffix trees~\cite{Baker96}]
The \emph{parameterized suffix tree} for a p-string $T$, denoted $\PSTree(T)$, is a compact trie that stores the set $\{\prev(T[i:]) \mid 1 \leq i \leq |T|\}$ of the previous encodings of all suffixes of $T$.
\end{definition}
See Figure~\ref{fig:PSTree} for examples of $\PSTree(T)$.
We assume that the leaves of $\PSTree(T)$ are sorted in lexicographical order,
so that the sequence of the leaves corresponds to
the parameterized suffix array for $T$, which is defined below.

  \begin{figure}[htbp]
  \begin{center}
    \includegraphics[width=90mm]{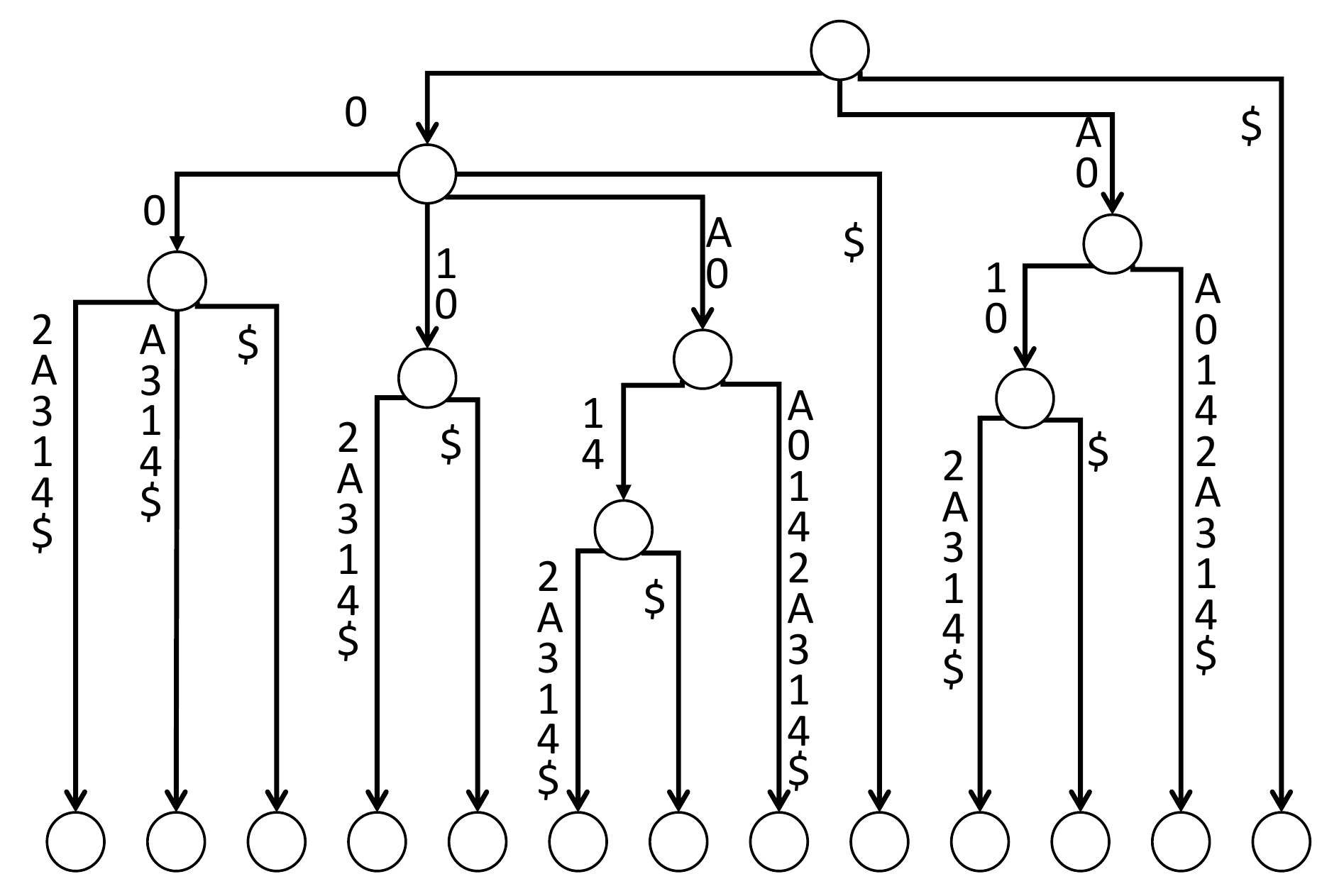}
  \end{center}
  \caption{$\PSTree(T)$ for a p-string $T=\mathtt{zAxAyyxyAxxy}$, where $\Sigma=\{\mathtt{A}, \mathtt{\$}\},\Pi=\{\mathtt{x},\mathtt{y}, \mathtt{z}\}$.}
  \label{fig:PSTree}
  \end{figure}

\begin{definition}[Parameterized suffix arrays~\cite{DeguchiHBIT08}]
The \emph{parameterized suffix array} of a p-string $T$, denoted $\PSA(T)$, is an array of integers such that $\PSA(T)[i] = j$ if and only if $\prev(T[j:])$ is the $i$th lexicographically smallest string in $\{ \prev(T[i:])\ |\ 1 \leq i \leq |T|\}$.
\end{definition}

\begin{definition}[Parameterized longest common prefix arrays~\cite{DeguchiHBIT08}]
  The parameterized longest common prefix array of a p-string $T$, denoted $\PLCP(T)$, is an array of integers such that $\PLCP(T)[1] = 0$ and $2 \leq i \leq |T|$ $\PLCP(T)[i]$ stores the length of the longest common prefix between $\prev(T[\PSA(T)[i-1]:])$ and $\prev(T[\PSA(T)[i]:])$.
\end{definition}
See Figure~\ref{fig:PSA_PLCP} for examples of $\PSA(T)$ and $\PLCP(T)$.

\begin{figure}
\begin{center}
\begin{tabular}{|c|c|l|c|}\hline
  $i$ & $\PSA(T)[i]$ & $T[\PSA(T)[i]:]$ & $\PLCP(i)$\\\hline
  1 & 6 & 0 0 2 $\mathtt{A}$ 3 1 4 \$ & 0\\\hline
  2 & 7 & 0 0 $\mathtt{A}$ 3 1 4 \$ & 2\\\hline
  3 & 11 & 0 0 \$ & 2\\\hline
  4 & 5 & 0 1 0 2 $\mathtt{A}$ 3 1 4 \$ & 1\\\hline
  5 & 10 & 0 1 0 \$ & 3\\\hline
  6 & 3 & 0 $\mathtt{A}$ 0 1 4 2 $\mathtt{A}$ 3 1 4 \$ & 1\\\hline
  7 & 8 & 0 $\mathtt{A}$ 0 1 4 \$ & 5\\\hline
  8 & 1 & 0 $\mathtt{A}$ 0 $\mathtt{A}$ 0 1 4 2 $\mathtt{A}$ 3 1 4 \$ & 3\\\hline
  9 & 12 & 0 \$ & 1\\\hline
  10 & 4 & $\mathtt{A}$ 0 1 0 2 $\mathtt{A}$ 3 1 4 \$ & 0\\\hline
  11 & 9 & $\mathtt{A}$ 0 1 0 \$ & 4\\\hline
  12 & 2 & $\mathtt{A}$ 0 $\mathtt{A}$ 0 1 4 2 $\mathtt{A}$ 3 1 4 \$ & 2\\\hline
  13 & 13 & \$ & 0\\\hline
  \end{tabular}
\end{center}
\caption{$\PSA(T)$ and $\PLCP(T)$ for a p-string $T=\mathtt{zAxAyyxyAxxy}$, where $\Sigma=\{\mathtt{A}, \mathtt{\$}\},\Pi=\{\mathtt{x},\mathtt{y}, \mathtt{z}\}$.}
  \label{fig:PSA_PLCP}
\end{figure}

In addition to the above data structures from the literature,
we introduce the following new notions and data structures.
For convenience, we will sometimes identify each node of
the parameterized suffix tree with the string
which is represented by that node.

In what follows,
let $\Pi_T = \{T[i] \in \Pi \mid 1 \leq i \leq |T|\}$
and $\Sigma_T = \{T[i] \in \Sigma \mid 1 \leq i \leq |T|\}$,
namely,
$\Pi_T$ (resp. $\Sigma_T$) is the set of distinct characters of $\Pi$ (resp. $\Sigma$) that occur in $T$.
Let $\pi = |\Pi_T|$ and $\sigma = |\Sigma_T|$.

\begin{definition}[P-nodes, branching p-nodes]
\label{pnode}
Let $T$ be a p-string over $\Sigma \cup \Pi$. A node $v$ in $\PSTree(T)$ is called a \emph{p-node} if the number of leaves in the subtree of $\PSTree(T)$ rooted at $v$ is at least $\max\{\sigma,\pi\}$.
A p-node $v$ is called a \emph{branching p-node} if at least two children of $v$ in $\PSTree(T)$ are p-nodes.
\end{definition}

See Figure~\ref{fig:p-node} for examples of p-nodes and branching p-nodes.

  \begin{figure}[htbp]
  \begin{center}
    \includegraphics[width=90mm]{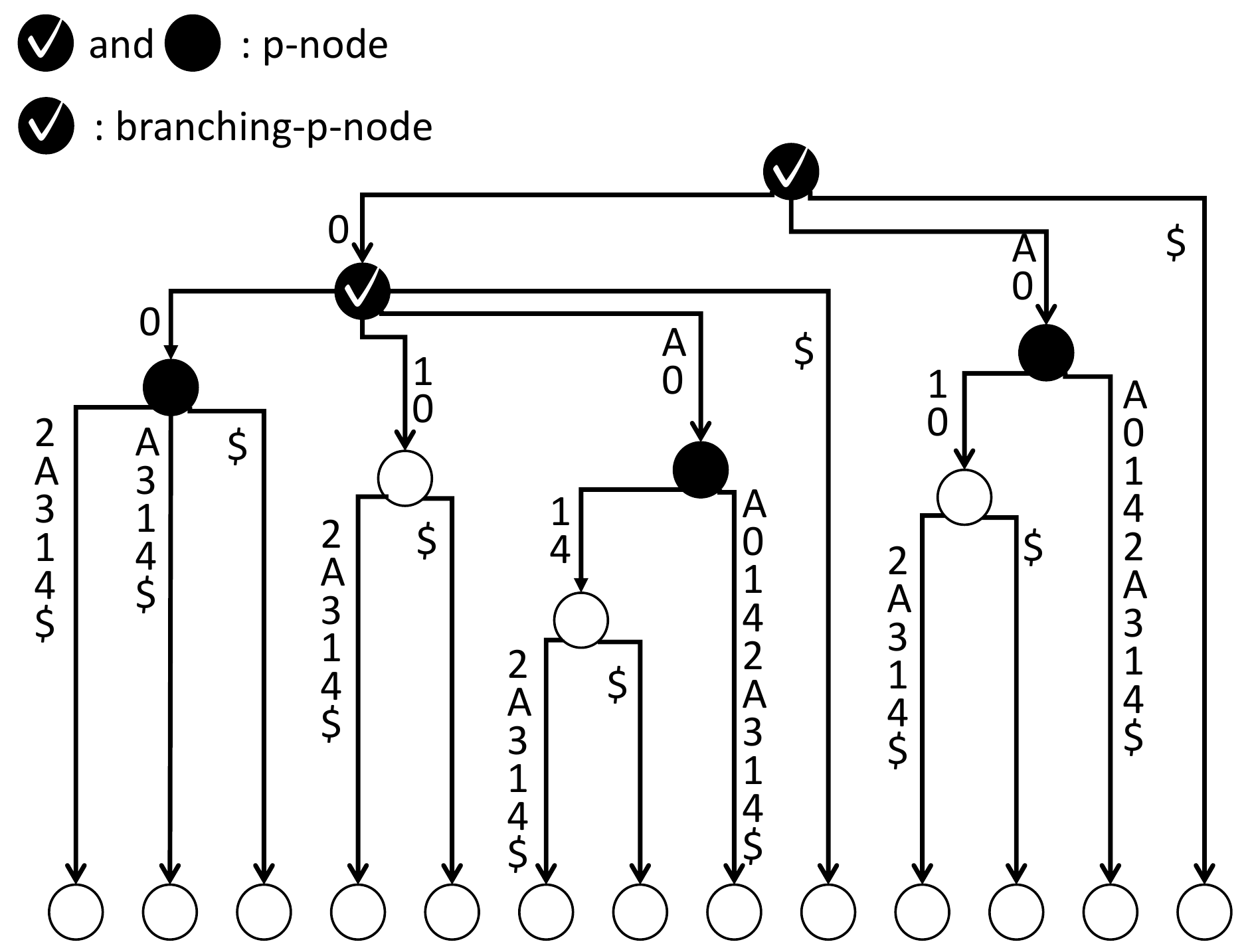}
  \end{center}
  \caption{$\PSTree(T)$ for a p-string $T=\mathtt{zAxAyyxyAxxy}$, where $\Sigma=\{\mathtt{A}, \mathtt{\$}\},\Pi=\{\mathtt{x},\mathtt{y}, \mathtt{z}\}$.Then black nodes are p-nodes because the number of leaves in the subtree of $\PSTree(T)$ rooted at them are at least $\max\{\sigma,\pi\}=3$. Checked nodes are branching p-nodes because at least two children of them in $\PSTree(T)$ are p-nodes.}
  \label{fig:p-node}
  \end{figure}

For any $x \in \Pi_T$, let $\rank_T(x)$ denote
the lexicographical rank of $x$ in $\Pi_T \cup \Sigma_T$.
Assuming that $\Pi$ and $\Sigma$ are integer alphabets of polynomial size in $n$,
we can compute $\rank_T(x)$ for every $x \in \Pi_T$ in $O(n)$ time
by bucket sort.
We will abbreviate $\rank_T(x)$ as $\rank(x)$ when it is not confusing.
  
\begin{definition}[P-array]
  Let $\prev(v)$ be any branching p-node of $\PSTree(T)$, where $v$ is some substring of $T$.
  The \emph{p-array} $A(\prev(v))$ for $\prev(v)$ is an array of length $\sigma+\pi$ such that
  for each $x \in \Sigma \cup \Pi$,
  $A(\prev(v))[\rank(x)]$ stores a pointer to the child $u$ of $\prev(v)$
  such that $\prev(\spe(v)x)$ is a prefix of $u$
  if such a child exists,
  and $A(\prev(v))[\rank(x)]$ stores nil otherwise.
\end{definition}

See Figure~\ref{fig:p-array} for an example of a p-array.

  \begin{figure}[htbp]
  \begin{center}
    \includegraphics[width=90mm]{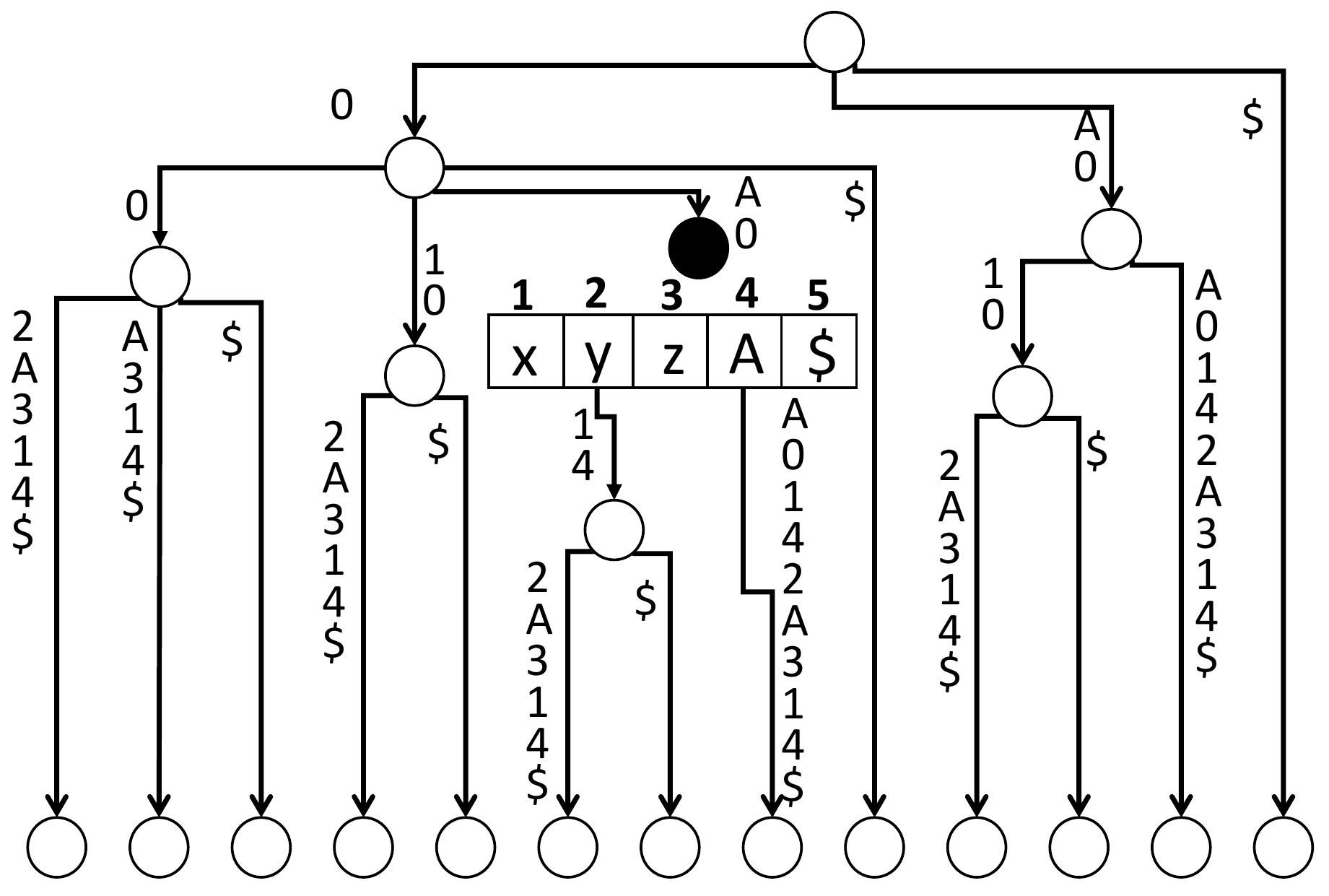}
  \end{center}
  \caption{$\PSTree(T)$ for a p-string $T=\mathtt{zAxAyyxyAxxy}$, where $\Sigma_T=\{\mathtt{A}, \mathtt{\$}\},\Pi_T=\{\mathtt{x},\mathtt{y}, \mathtt{z}\}$. Consider a branching p-node $\prev(v) = \mathtt{0A0}$ where $v$ is e.g. $\mathtt{zAx}$. Then, $A(\mathtt{0A0})[\rank(\mathtt{y})] = A(\mathtt{0A0})[2]$ stores a pointer to node $\mathtt{0A014}$ because $\spe(v)=\mathtt{xAy}$ and $\prev(\spe(v)\mathtt{y}=\mathtt{xAyy})=\mathtt{0A01}$ is a prefix of node $\mathtt{0A014}$.}
  \label{fig:p-array}
  \end{figure}

\begin{definition}[Parameterized suffix tray]
  The parameterized suffix tray of a p-string $T$, denoted $\PSTray(T)$, is a hybrid data structure consisting of $\PSA(T)$, $\PLCP(T)$, and $\PSTree(T)$ where each branching p-node is augmented with the p-array.
\end{definition}

We can show the following lemmas regarding the space requirements of
$\PSTray(T)$, by similar arguments to~\cite{ColeTL2015}
for suffix trays on standard strings.

\begin{lemma} 
For any p-string $T$ of length $n$ over $\Sigma \cup \Pi$,
the number of branching p-nodes in $\PSTree(T)$ is $O(\frac{n}{\pi+\sigma})$.
\end{lemma}

\begin{lemma} 
For any p-string $T$ of length $n$,
$\PSTray(T)$ occupies $O(n)$ space.
\end{lemma}

\section{PPM using parameterized suffix trays}

In this section, we present our algorithm for parameterized pattern matching (PPM) on $\PSTray(T)$.
For any node $v$ in $\PSTray(T)$,
let $l_v=(i,j)$ denote the range of $\PSA(T)$ that $v$ corresponds,
namely, $l_v=(i,j)$ iff
the leftmost and rightmost leaves in the subtree rooted at $v$
correspond to the $i$th and $j$th entries of $\PSA(T)$,
respectively.
For any \emph{p-node} $v$, we store $l_v$ in $v$.
Also, for any \emph{non-branching} p-node $u$,
we store a pointer to the unique child of $u$ that is a p-node.
These can be easily computed in a total of $O(n)$ time
by a standard traversal on $\PSTree(T)$.

The basic strategy for PPM with $\PSTray(T)$ follows
the (exact) pattern matching algorithm with suffix trays on standard strings~\cite{ColeTL2015}.
Namely, we traverse $\PSTree(T)$ with a given pattern $P$ from the root,
and as soon as we encounter a node that is not a p-node,
then we switch to the corresponding range of $\PSA(T)$ and perform
a binary search to locate the pattern occurrences.
The details follow.

Let $P$ be a pattern p-string of length $m$.
We assume that $\Pi$ and $\Sigma$ are disjoint integer alphabets, where $\Pi = \{0,\ldots,c_1n\}$ and $\Sigma = \{c_1n +1,\ldots,n^{c_2}\}$ for some positive constants $c_1$ and $c_2$.
Using an array (bucket) $B$ of size $|\Pi| = c_1n \in O(n)$,
we can compute $\prev(P)$ in $O(m)$ time
by scanning $P$ from left to right and
keeping the last occurrence of each character $x \in \Pi$ in $P$ in $B[x]$.
We can compute $\spe(P)$ in $O(m)$ time in a similar manner with a bucket.
These buckets are a part of our indexing structure
that occupies $O(n)$ total space.

After computing $\prev(P)$ and $\spe(P)$, we traverse $\prev(P)$ on $\PSTray(T)$.
If $\prev(P[: i])$ for prefix $P[:i]$~($1\leq i \leq m$) is represented by a p-node, we can find the out-going edge whose label begins with $\prev(P)[i+1]$  in constant time by accessing the p-array entry $A(\prev(P[:i]))[\spe(P)[i + 1]]$. Therefore, we can solve PPM in $O(m + \occ)$ time if $\prev(P)$ is a prefix of some p-node in $\PSTray(T)$.
Otherwise (if $\prev(P)$ is not a prefix of any p-node),
there exists integer $i$ such that
$\prev(P[: i])$ is not a p-node but the parent of $\prev(P[: i])$ is a p-node.
In this case, we will use the next lemma.
\begin{lemma}[PPM in PSA range (adapted from~\cite{DeguchiHBIT08})]
  \label{psa}
  Given a pattern p-string $P$ of length $m$ and a range $[j,k]$ in $\PSA(T)$
  such that the $\occ$ occurrences of $P$ in $T$ lie in the range $[j,k]$ of $\PSA(T)$, we can find them in $O(m+\log (k-j)+\occ)$ time by using $\PSA(T)$ and $\PLCP(T)$.
\end{lemma}
Let $I_{\prev(P[:i])}=(j,k)$ denote the range in $\PSA(T)$
where $\prev(P)$ is a prefix of the suffixes in the range.
We apply Lemma~\ref{psa} to this range so we can find the
parameterized occurrences of $P$ in $T$ in $O(m + \log(k-j) + \occ)$ time.
By Definition~\ref{pnode} we have $k-j \leq \pi+\sigma$
(recall that $\prev(P[: i])$ is not a p-node).
Thus, $O(m + \log(k-j) + \occ) \subseteq O(m + \log(\pi+\sigma) + \occ)$,
implying the next theorem.

\begin{theorem}
  Suppose $|\Pi| = O(n)$.
  Then, $\PSTray(T)$ supports PPM queries in $O(m + \log(\pi+\sigma) + \occ)$ time each, where $m$ is the length of a query pattern $P$ and $\occ$ is the number of occurrences to report.
\end{theorem}

\section{Construction of parameterized suffix trays}

Let $T$ be a p-string of length $n$.
In this section, we show how to construct $\PSTray(T)$
provided that $\PSTree(T)$ has already been built.
Throughout this section
we assume that $\Pi$ and $\Sigma$ are disjoint integer alphabets,
both being of polynomial size in $n$, namely,
$\Pi = \{0,...,n^{c_1}\}$ and $\Sigma = \{n^{c_1} +1,...,n^{c_2}\}$ for some positive constants $c_1$ and $c_2$.
For convenience, we define the following two notions.

\begin{definition}[P-function]
  \label{function}
  Let $q,r$ be p-strings such that $q \approx r$.
  The \emph{p-function} $\f_{q,r} : \Sigma\cup\Pi \rightarrow \Sigma\cup\Pi$
  transforms $q$ to $r$, namely, for every $1 \leq i \leq h$
  \[
    \f_{q,r}(q[i])=r[i].
  \]
\end{definition}

For instance, if $\Pi =\{\mathtt{x},\mathtt{y},\mathtt{z}\}$, $q=\mathtt{xyxzyyxz}$, $r=\mathtt{zxzyxxzy}$ and $q\approx r$, then $\f_{q,r}(\mathtt{x})=\mathtt{z}$, $\f_{q,r}(\mathtt{y})=\mathtt{x}$, $\f_{q,r}(\mathtt{z})=\mathtt{y}$ since $q$ can be transformed $r$ by this function.

\begin{definition}[F-array]
  Let $q$ be a p-string and $x \in \Pi_T$. The first (left-most) occurrence of $x$ in $q$ is denoted by $i_{q,x}$.
  The \emph{f-array} of $q$, denoted $\fpos(q)$, is an array of length $\pi$ such that $\fpos(q)[\rank(x)] = i_{q,x}$.
\end{definition}

For instance, if $\Pi_T =\{\mathtt{x},\mathtt{y},\mathtt{z}\}$ and $q=\mathtt{xyxzyyxz}$, then $\fpos(q)[\rank(\mathtt{x})] = \fpos(q)[1] =1$, $\fpos(q)[\rank(\mathtt{y})] = \fpos(q)[2] = 2$, and $\fpos(q)[\rank(\mathtt{z})] = \fpos(q)[3] \\ = 4$.

Given $\PSTree(T)$, we show how to construct $\PSTray(T)$. It is well known that $\PSA(T)$ and $\PLCP(T)$ can be constructed from $\PSTree(T)$ in $O(n)$ time.
In the following, we consider how to compute $A(\prev(v))$ for every p-node $\prev(v)$ in $\PSTree(T)$.

First, we consider how to compute (branching) p-nodes in $\PSTree(T)$. This can be done by a similar method to the suffix tray for exact matching~\cite{ColeTL2015}, namely:
\begin{lemma}[Computing p-node]
  We can compute all p-nodes and branching p-nodes in $\PSTree(T)$ in $O(n)$ total time.
\end{lemma}


Our algorithm performs a bottom-up traversal on $\PSTree(T)$
and propagates pairs $(\fpos(T[i:]), i)$ from leaves to their ancestors.
Each internal p-node $\prev(v)$ will store only a single pair
$(\fpos(T[i:]), i)$, where
$i$ is the largest position in $T$
such that $\prev(T[i:])$ is a leaf in the subtree rooted at $\prev(v)$
\footnote{Indeed, our $\PSTray(T)$ construction algorithm works with any position $i$ in the subtree rooted at $\prev(v)$, and we propagate the largest leaf position $i$ to each internal p-node for simplicity.}.
See also Figure~\ref{fig:fpos}.
One can easily compute the pairs for all p-nodes in a total of $O(n)$ time.
\begin{figure}[htbp]
  \begin{center}
    \includegraphics[width=110mm]{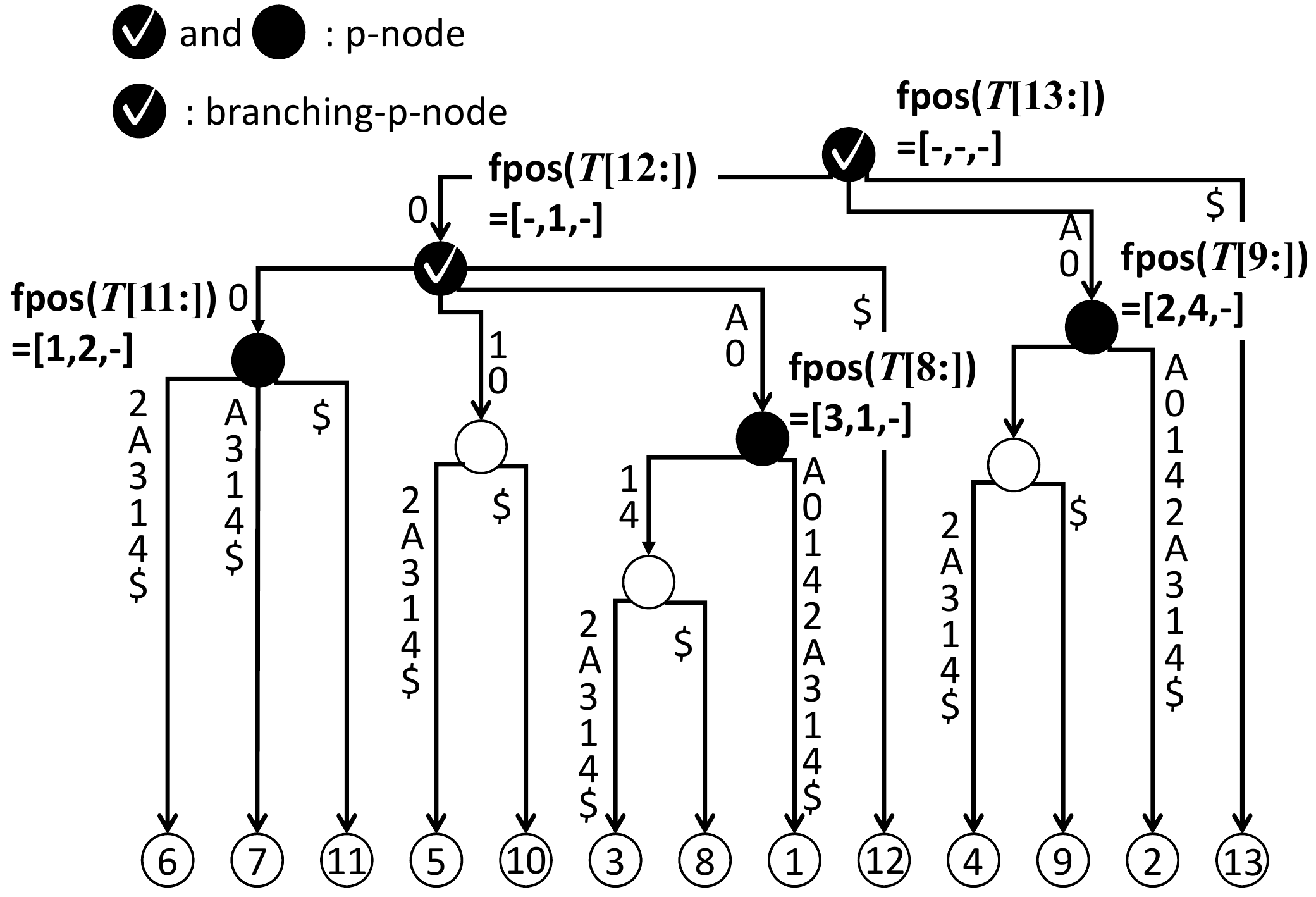}
  \end{center}
  \caption{$\PSTree(T)$ for a p-string $T=\mathtt{zAxAyyxyAxxy}$, where $\Sigma=\{\mathtt{A}, \mathtt{\$}\}$ and $\Pi=\{\mathtt{x},\mathtt{y}, \mathtt{z}\}$. For instance, we propagate $\fpos(T[12:])$ (coupled with the corresponding position $12$) to the p-node $\mathtt{0}$.}
  \label{fig:fpos}
\end{figure}
Then, we compute $\f_{v,\spe(v)}$ for every p-node $\prev(v)$ from
the pair $(\fpos(T[i:]), i)$ that is stored in the p-node $\prev(v)$.
Finally, for every p-node $\prev(v)$ we compute $A_{\prev(v)}$ from $\f_{v,\spe(v)}$ and $i$.

In what follows, we first show how to compute $A_{\prev(v)}$ from $\f_{v,\spe(v)}$ and $i$ in Lemmas~\ref{lem:spe_equality} and ~\ref{lem:f_to_spe}.
We then present how to compute $\f_{v,\spe(v)}$ and $i$ from $\fpos(T[i:])$ in Lemma~\ref{pitime},
and how to compute $\fpos(T[i:])$ in Lemma~\ref{newfpos}.
These lemmas will ensure the correctness and time complexity of our algorithm.

We consider how to compute $A(\prev(v))$ for a given p-node $\prev(v)$.

\begin{lemma} \label{lem:spe_equality}
  Let $s$ be a p-string. If $\prev(s)[|s|] = k \in \{0, \ldots, |T|-1\}$, then
  $\spe(s)[|s|] = \spe(s)[|s|-k]$.
\end{lemma}

\begin{proof}
  Clear from the definitions of $\prev(\cdot)$ and $\spe(\cdot)$.
\end{proof}

In the sequel,
let $\prev(T[i:])$ be any leaf in the subtree rooted at $\prev(v)$, where $1 \leq i \leq |T|$.
By Lemma~\ref{lem:spe_equality},
we can compute $A(\prev(v))$ if we know $\spe(v)[|v|-k + 1]$,
where $k = \prev(T[i:])[|v|+1]$.

\begin{lemma}\label{lem:f_to_spe}
  $\spe(v)[|v|-k+1]=\f_{T[i:i+|v|-1],\spe(T[i:i+|v|-1])T[i+|v|+k-2]}$.
\end{lemma}

\begin{proof}
  Clear from the definitions of $\f_{\cdot,\cdot}$ and $\spe(\cdot)$.
\end{proof}

We can compute $\spe(v)[|v|-k+1]$ if we know $\f_{T[i:i+|v|-1],\spe(v)}$.
In the next lemma, we show how to compute
$\f_{T[i:i+|v|-1],\spe(v)}$ for all p-nodes $\prev(v)$.

\begin{lemma}\label{pitime}
  For every p-node $\prev(v)$,
  we can compute $\f_{T[i:i+|v|-1],\spe(T[i:i+|v|-1])}$
  with $\prev(v) = \prev(T[i:i+|v|-1]))$
  in amortized $O(\pi)$ time if we know pair $\fpos(T[i:], i)$.
\end{lemma}

\begin{proof}
  Let $x_j$ be the $j$th smallest element of $\Pi$ in lexicographically order.
  Let $y_l$ denote the parameterized character in $S = T[i:i+|v|-1]$ such that
  if $j$ is the left-most occurrence of $y_l$ in $S$ (i.e. $j = \min\{h \mid S[h] = y_l\}$),
  then $|\Pi_{S[1..j]}| = l$.
  For instance, for $S = \mathtt{xAzAyA}$,
  then $y_1=\mathtt{x}$, $y_2=\mathtt{z}$, and $y_3=\mathtt{y}$.
  Let $(i_{\prev(v)}, \fpos(T[i_{\prev(v)}:])$ denote the pair stored in
  p-node $\prev(v)$.
  Consider a set
  $\mathbf{I} = \{(i_{\prev(v)}, \fpos(T[i_{\prev(v)}:])[k]) \mid \text{$\prev(v)$ is a p-node}, 1 \leq k \leq \pi\}$ of integer pairs.
  We sort the elements of $\mathbf{I}$
  so that we can compute in $O(1)$ time
  $\f_{T[i:i+|v|-1], \spe(v)}(y_l)= x_l$
  for all p-nodes $\prev(v)$, where $x_l$ is the $l$th smallest
  parameterized character that occurs in $\spe(v)$.
  We can sort the elements of $\mathbf{I}$ 
  in a total of $O(n)$ time by radix sort,
  since there are $O(\frac{n}{\sigma + \pi})$ p-nodes
  and each f-array $\fpos(T[i:])$ is of length $\pi$.  
  This completes the proof.
\end{proof}

We can easily compute $\fpos(T[i :])$ by the following lemma:
\begin{lemma}
  \label{newfpos}
  Let $q$ be a p-string. Let $x\in\Pi$, $y\in\Pi\cup\Sigma$ and $x\neq y$. Then the following equations hold:
  \begin{eqnarray*}
    \fpos(xq)[\rank(x)] & = & 1, \\
    \fpos(xq)[\rank(y)] & = & \fpos(q)[\rank(y)]+1.
  \end{eqnarray*}
\end{lemma}
Thus we can compute f-arrays $\fpos(T[i:])$ for all $1 \leq i \leq n$
in a total of $O(n)$ time.

Since all the afore-mentioned procedures take $O(n)$ time each,
we obtain the main theorem of this section.
\begin{theorem}
  Given a p-string $T$ of length $n$ over alphabet $\Sigma \cup \Pi$
  with $\Pi = \{0,...,n^{c_1}\}$ and $\Sigma = \{n^{c_1} +1,...,n^{c_2}\}$ for some positive constants $c_1$ and $c_2$,
  we can construct $\PSTray(T)$ in $O(n)$ time from $\PSTree(T)$.
\end{theorem}

\section{Conclusions and open questions}

In this paper, we proposed an indexing structure for
parameterized pattern matching (PPM)
called the parameterized suffix tray $\PSTray(T)$,
where $T$ is a given text string.
Our $\PSTray(T)$ uses $O(n)$ space and
supports pattern matching queries in
$O(m + \log (\sigma + \pi) + \occ)$ time,
where $n = |T|$,
$m$ is the query pattern length,
$\sigma$ and $\pi$ are respectively the numbers of
distinct static characters and distinct parameterized characters
occurring in $T$,
and $\occ$ is the number of pattern occurrences to report.
We also showed how to construct $\PSTray(T)$ in $O(n + s(n))$ time,
where $s(n)$ denotes the time complexity to build
the parameterized suffix tree $\PSTree(T)$ for $T$.
It is known that $s(n) =\min \{n\pi , n(\log(\pi + \sigma)\}$~\cite{Baker96,Shibuya04,FujisatoNIBT19}.

On the other hand, if we use hashing for implementing
the branches of the parameterized suffix tree $\PSTree(T)$,
one can trivially answer PPM queries in $O(m+\occ)$ time with $O(n)$ space.
The best linear-space deterministic hashing we are aware of 
is the one by Ru\v{z}i\'{c}~\cite{Ruzic08},
which can be built in $O(n (\log \log n)^2)$ time for a set of $n$ keys
in the word RAM model with machine word size $\Omega(\log n)$.
By associating each node of $\PSTree(T)$ with a unique integer
(e.g. the pre-order rank),
one can regard each branch in $\PSTree(T)$ as an integer from the universe
of polynomial size in $n$,
each fitting in a constant number of machine words.
This gives us a deterministic $O(n (\log \log n)^2 + s(n))$-time
algorithm for building $\PSTree(T)$ with $O(m + \occ)$-time PPM queries.
Still, it is not known whether a similar data structure
can be build in $O(n + s(n))$ time.
We conjecture that our $O(m + \log(\sigma + \pi) + \occ)$ PPM query time
would be the best possible for any indexing structure
that can be build in $O(n + s(n))$ time.
Proving or disproving such a lower bound is an intriguing open problem.

\section*{Acknowledgments}
This work was supported by JSPS KAKENHI Grant Numbers
JP18K18002 (YN),
JP17H01697 (SI),
JP20H04141 (HB), JP18H04098 (MT),
and JST PRESTO Grant Number JPMJPR1922 (SI).

\bibliographystyle{abbrv}
\bibliography{ref}

\end{document}